\documentclass[reqno]{amsart}
\usepackage{amsmath}
\usepackage{amsfonts}
\usepackage{amssymb}
\usepackage{amssymb,graphics,graphicx,bbm,hyperref,color}

\usepackage{multicol}
\usepackage{enumitem} 

\definecolor{gray}{rgb}{0.93,0.93,0.93}
\definecolor{light-gold}{rgb}{0.99,0.97,0.78}

\def\be{\begin{equation}}
\def\ee{\end{equation}}
\newcommand{\eq}[1]{(\ref{#1})}
\def\bm{\begin{multline}}
\def\bfig{\begin{figure}[htb]}
\def\efig{\end{figure}}

\newcommand{\e}[1]{\,{\rm e}^{#1}\,}
\newcommand{\ii}{{\rm i}}
\def\Tr{{\operatorname{Tr\,}}}

\newcommand{\sumtwo}[2]{\sum_{\substack{#1 \\ #2}}}

\DeclareMathOperator*{\Union}{\text{\Large$\cup$}}
 \def\idty{{\mathchoice {\mathrm{1\mskip-4mu l}} {\mathrm{1\mskip-4mu l}} %
{\mathrm{1\mskip-4.5mu l}} {\mathrm{1\mskip-5mu l}}}}

\numberwithin{equation}{section}
\newtheorem{theorem}{Theorem}[section]
\newtheorem{proposition}[theorem]{Proposition}
\newtheorem{lemma}[theorem]{Lemma}

\theoremstyle{remark}

\newcommand{\caH}{{\mathcal H}}
\newcommand{\caL}{{\mathcal L}}

\newcommand{\bbC}{{\mathbb C}}
\newcommand{\bbE}{{\mathbb E}}
\newcommand{\bbN}{{\mathbb N}}
\newcommand{\bbP}{{\mathbb P}}

\newcommand{\Rl}{{\mathbb R}}

\newcommand{\bbZ}{{\mathbb Z}}
\newcommand{\Ir}{{\mathbb Z}}

\makeatletter
  \def\tagform@#1{\maketag@@@{\footnotesize{(#1)}\@@italiccorr}}
\makeatother

\renewcommand{\eqref}[1]{(\ref{#1})}

\begin{document}


\title[A proof of dimerization in $SU(n)$-invariant quantum spin chains]{A direct proof of dimerization in a family of $SU(n)$-invariant quantum spin chains}

\author{Bruno Nachtergaele}
\address{Department of Mathematics, University of California,
One Shields Ave, Davis, CA 95616, United States}
\email{bxn@math.ucdavis.edu}

\author{Daniel Ueltschi}
\address{Department of Mathematics, University of Warwick,
Coventry, CV4 7AL, United Kingdom}
\email{daniel@ueltschi.org}



\begin{abstract}
We study the family of spin-$S$ quantum spin chains with a nearest neighbor interaction given by the negative of the 
singlet projection operator. Using a random loop representation of the partition function in the limit of zero temperature 
and standard techniques of classical statistical mechanics, we prove dimerization for all sufficiently large values of  $S$.
\end{abstract}

\thanks{\copyright{} 2017 by the authors. This paper may be reproduced, in its entirety, for non-commercial purposes.}

\maketitle

\section{Introduction}
\label{sec intro}

Dimerization is the most common type and the most common mechanism of spontaneous lattice translation symmetry 
breaking in the ground state of quantum spin systems. It is ubiquitous in one dimension due to the so-called (spin-) 
Peierls instability \cite{cross:1979}. In two and more dimensions it gives rise to columnar phases and other 
patterns of lattice symmetry breaking \cite{sirker:2002,ralko:2008,frank:2011,giuliani:2015}. In the ground states of 
quantum spin chains with isotropic antiferromagnetic interactions, long-range anti-ferromagnetic order, and the 
accompanying spontaneous breaking of the continuous rotation symmetry, is prevented by quantum fluctuations. 
Short-range antiferromagnetic correlations, however, occur in several manifestations, some of which involve discrete 
symmetry breaking. Due to a result by Affleck and Lieb \cite{affleck:1986}, half-integer spin chains with a a rotation 
symmetry and a unique ground state, have a gapless excitation spectrum above the ground state. The opening up of a 
non-vanishing (in the thermodynamic limit) gap above the ground state requires that the ground state is degenerate 
and for a class of reflection positive antiferromagnetic chains it is known that the nature of this degeneracy is dimerization, 
i.e., breaking of the translation invariance from $\Ir$ to $2\Ir$ \cite{AN,nachtergaele:1994}.

In this paper we study a family of spin-$S$ chains with an isotropic nearest neighbor Hamiltonian for a chain of $L$ spins given by 
\be
H(S,L)= -\sum_{x=1}^{L-1} P_{x,x+1}^{(0)},
\label{HamSN}\ee
where $P_{x,x+1}^{(0)}$ denotes the orthogonal projection onto the singlet state of the two spins at sites $x$ and $x+1$.
Since $P_{x,x+1}^{(0)}$ commutes with the tensor product of the fundamental and the anti-fundamental representation 
of $SU(2S+1)$, this model has an $SU(2S+1)$ symmetry. For $S=1/2$ this Hamiltonian is, up to a constant, the standard 
spin-$1/2$  antiferromagnetic Heisenberg chain and its Bethe ansatz gives a unique, translation-invariant ground 
state with power law correlations and a gapless excitation spectrum \cite{korepin:1993}. 
For $S\geq 1$ the ground state is expected to be dimerized and two-fold degenerate, and with a positive spectral gap.
This expectation is supported by relationships between the spin chain Hamiltonian $H(S,L)$ and other spin Hamiltonians that 
are diagonalizable using the Bethe Ansatz or other exact solution methods \cite{barber:1989,klumper:1990,affleck:1990}. 
E.g., the $S=1$ model can be related to spin-$1/2$ $XXZ$ chain with anisotropy $\Delta = -3/2$, and to the transfer matrix 
of the standard $9$-state Potts model on a square lattice. These relationships stem from the observation that $H(S,L)$ and 
the related spin-chain Hamiltonian or transfer matrix are both representatives of an element $H_N$ in the abstract Temperley-Lieb 
algebra $TL_N(x)$, with $x=2S+1$. This implies that the spectrum of each of these operators is a subset of the algebraic 
spectrum of $H_N$, i.e., the complex numbers $\lambda$ for which $H_N -\lambda\idty$ fails to have an inverse 
in $TL_N(x)$. Many algebras have only one irreducible representation, in which case the spectrum is the same in any 
representation with only the multiplicity remaining to be determined. The number of irreducible representations of the 
Temperley-Lieb algebra $TL_N(x)$, however, grows as $\sim N^2$ and two representations could therefore, in principle, give 
entirely different spectra to the element $H_N$. Nevertheless, exact calculations on small chains and numerical calculations 
indicate that the spectra coincide up to multiplicity. More recent results on the algebraic Temperley-Lieb chains may point to 
an explanation of this `universality' of the spectrum. See \cite{nepomechie:2016} and the references therein for a discussion of
this phenomenon.

The relationship between the spin-$S$ chain with Hamiltonian \eq{HamSN} with $S\geq 1$
and the 2-dimensional $q$-state Potts model with $q=(2S+1)^2$ at the self-dual point extends in a non-trivial way 
to the states: The two dimerized ground states correspond to the coexisting ordered and disordered phases of
the $q$-state Potts-model on the square lattice at its critical point. This coexistence (a first order phase transition) has been 
found by Baxter \cite{baxter:1973} for $q\geq 5$ and has been rigorously established for all sufficiently large values of $q$ 
\cite{kotecky:1982}. Even if one accepts the identity of the spectra of the related Hamiltonians, relating order parameters or, 
in general, the states of these models is a subtle issue. In \cite{AN} it was shown that the dimerization order parameter for 
the spin chain and the order parameter of a suitable two-dimensional ferromagnetic Potts model are bounded by a multiple 
of each other. This implies that one of them vanishes if and only if the other does. The Potts model for which 
this relationship is proved can be regarded as a strongly anisotropic limit (with respect to the two lattice directions), in which 
one of the lattice directions tends to the continuum. A similar result likely also holds for the standard Potts model on the square 
lattice, but this has not been worked out in the literature.

Due to the subtleties of the relationships between the different models discussed above, a complete proof of dimerization has
been lacking. In view of the non-trivial nature of rigorous studies of the critical Potts model itself, it seemed worthwhile to look 
for a direct proof of the dimerization, bypassing establishing more details of the relationship between the various models.
In this article we provide the first complete proof of the existence of dimerized ground states for this class of models at sufficiently
large values of $S$. Moreover, by not relying on the Temperley-Lieb algebra and a relation to the Potts models, we open the 
possibility to study perturbations of the model away from the self-dual point in the phase diagram. This is important in the context 
of the current interest in stable gapped ground state phases of quantum lattice systems, which we briefly discuss in 
Section \ref{sec discussion}.

Our approach is based on a random loop representation for the partition function of the spin models \eq{HamSN} given in 
\cite{AN}, which has been applied in recent years for a number of other rigorous results for quantum spin models
\cite{BU,CNS,Lees}. For a review and extensions of the random loop representation for quantum spin models see \cite{Uel}. 
We give a precise statement of our results in the next section. A detailed description of the random loop representation and its 
properties in Section \ref{sec loop rep}. In Section \ref{sec contour rep} we introduce a suitable notion of contours and the Peierls
argument proof of our main result. We conclude with a short discussion.

\section{Setting \& results}
\label{sec setting}

For $\ell \in \bbN$, consider a chain of even length $2\ell$ consisting of the set of vertices $\{-\ell+1, -\ell+2, \dots, \ell\}$ with edges 
between nearest-neighbors. For $S\in  \frac12 \bbN$, the Hilbert space of a spin-$S$ chain is then
\be
\caH_\ell = \otimes_{x=-\ell+1}^\ell \bbC^{2S+1},
\ee
The interaction is nearest-neighbor and given by the orthogonal projection $P_{x,x+1}^{(0)}$ onto the spin singlet at the 
sites $x$ and $x+1$. In terms of the standard tensor product basis of  $\bbC^{2S+1} \otimes \bbC^{2S+1}$ constructed 
with the eigenvectors on the third component of the spin, this operator is given by
\be
P_{x,y}^{(0)} = \frac1{2S+1} \sum_{a,b=-S}^S (-1)^{a-b} |a,-a\rangle \langle b,-b|.
\ee
We also write $P_{x,x+1}^{(0)}$ for the operator on $\caH_\ell$, obtained by tensoring it with the identity on the Hilbert spaces 
corresponding to sites other than $x$ and $x+1$. The Hamiltonian of the models studied in this paper is
\be
H_\ell = -\sum_{x=-\ell+1}^{\ell-1} P_{x,x+1}^{(0)}.
\ee

The interaction can be written in terms of the usual spin operators $S_x^i$, $i=1,2,3$, which satisfy
\be
[S_x^1, S_y^2] = \delta_{x,y} \ii S_x^3,
\ee
and further relations obtained by cyclic permutations of the indices $1,2,3$, and also satisfy
\be
(S_x^1)^2 + (S_x^2)^2 + (S_x^3)^2 = S(S+1).
\ee
For instance, for $S=1/2, 1$, and $3/2$, one has
\be
\begin{aligned}
&S=\tfrac12: & P_{x,y}^{(0)} &= - \vec S_x \cdot \vec S_y + \tfrac14, \\
&S=1: &  P_{x,y}^{(0)} &= \tfrac13 (\vec S_x \cdot \vec S_y)^2 - \tfrac13\\
&S=\tfrac32: &  P_{x,y}^{(0)} &= -\tfrac{1}{18}  (\vec S_x \cdot \vec S_y)^3  - \tfrac{5}{72}   (\vec S_x \cdot \vec S_y)^2
+ \tfrac{31}{96}\vec S_x \cdot \vec S_y + \tfrac{33}{128}.
\end{aligned}
\ee

The $S=\frac12$ model is the usual Heisenberg antiferromagnet and it is not expected to display long-range order in its one-dimensional ground 
state. Dimerization is expected for $S\geq1$ and in this paper we prove it for $S\geq 40$. More specifically, we derive a lower 
bound for a dimerization parameter order in the zero-temperature limit of Gibbs states on even chains, uniformly in the length of the chain.
To state our result, let $\langle\cdot\rangle_{\beta,\ell}$ denote the Gibbs state,
\be
\langle a \rangle_{\beta,\ell} = \frac1{\Tr \e{-\beta H_\ell}} \Tr a \e{-\beta H_\ell},
\label{gibbs}\ee
and let $\langle\cdot\rangle_{\infty,\ell} = \lim_{\beta\to\infty} \langle\cdot\rangle_{\beta,\ell}$ denote the expectation in the ground state. Our main result is the following theorem.

\begin{theorem}
\label{thm dimers}
For $S\geq40$, there exists $c(S)>0$\footnote{An explicit expression is derived in the proofs of this paper, in particular, the proof of
Proposition \ref{prop Peierls}. This yields the bound with 
$$
c(S) =\left(1-\frac{1}{(2S+1)^2}\right)\left(1-2\left[ \frac{64}{(2S+1)^{3/2}} + \frac{128}{(2S+1)^2} + \frac1{12} \sum_{k\geq7} \frac{(k+1) 4^k}{(2S+1)^{\frac12 k - 1}}
\right]\right).
$$ 
The geometric series is summable for $S>15/2$, which is the explanation why exponential decay of correlations (and hence dimerization) can be shown for 
$S\geq 8$ (Theorem \ref{thm exp decay}). Note that the expression for $c(S)$ is positive for $S\geq 40$.
} such that
\[
\langle P_{x,x+1}^{(0)} \rangle_{\infty,\ell} - \langle P_{x-1,x}^{(0)} \rangle_{\infty,\ell} > c(S),
\]
for all $x \in \{-\ell+3, -\ell+5, \dots,\ell-1\}$, uniformly in $\ell \in \bbN$.
\end{theorem}

The theorem states that, for any $x \in \{-\ell+3, -\ell+5, \dots,\ell-1\}$, the probability that the bond $\{x,x+1\}$ is in the singlet state (dimerized) 
exceeds the probability that $\{x-1,x\}$ is in the singlet state by a positive amount, uniformly in $\ell$. 
See Figure\ \ref{fig dimers} for an illustration. Thus, in the limit $\ell\to\infty$ along even values, one gets a ground state where $\{-1,0\}$ is more likely to be dimerized than $\{0,1\}$. In the limit $\ell\to\infty$ along odd values, the converse is true. This establishes the existence of two distinct, non-translation
invariant ground states.  By an averaging procedure, one sees that there are two periodic ground states of period 2.  We conjecture that these are the only ground states of the infinite chain. The proof of Theorem \ref{thm dimers} can be found at the end of Section \ref{sec contour rep}.

An open question is whether dimerization occurs in higher dimensions. The presence of N\'eel order has been proved for $S=\frac12$ and $d\geq3$ \cite{DLS}, and for any $S \geq 1$ provided the dimension is large enough (depending on $S$) \cite{Uel}. This leaves open the possibility that for fixed $d$, such as $d=2$, dimerization occurs when $S$ is large enough.

\begin{centering}
\bfig
\begin{picture}(0,0)%
\includegraphics{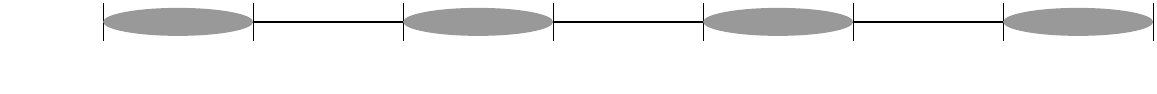}%
\end{picture}%
\setlength{\unitlength}{2368sp}%
\begingroup\makeatletter\ifx\SetFigFont\undefined%
\gdef\SetFigFont#1#2#3#4#5{%
  \reset@font\fontsize{#1}{#2pt}%
  \fontfamily{#3}\fontseries{#4}\fontshape{#5}%
  \selectfont}%
\fi\endgroup%
\begin{picture}(9247,736)(376,-2135)
\put(376,-2086){\makebox(0,0)[lb]{\smash{{\SetFigFont{7}{8.4}{\rmdefault}{\mddefault}{\updefault}{\color[rgb]{0,0,0}$-\ell+1=-3$}%
}}}}
\put(2251,-2086){\makebox(0,0)[lb]{\smash{{\SetFigFont{7}{8.4}{\rmdefault}{\mddefault}{\updefault}{\color[rgb]{0,0,0}$-2$}%
}}}}
\put(4726,-2086){\makebox(0,0)[lb]{\smash{{\SetFigFont{7}{8.4}{\rmdefault}{\mddefault}{\updefault}{\color[rgb]{0,0,0}$0$}%
}}}}
\put(7126,-2086){\makebox(0,0)[lb]{\smash{{\SetFigFont{7}{8.4}{\rmdefault}{\mddefault}{\updefault}{\color[rgb]{0,0,0}$2$}%
}}}}
\put(8326,-2086){\makebox(0,0)[lb]{\smash{{\SetFigFont{7}{8.4}{\rmdefault}{\mddefault}{\updefault}{\color[rgb]{0,0,0}$3$}%
}}}}
\put(3451,-2086){\makebox(0,0)[lb]{\smash{{\SetFigFont{7}{8.4}{\rmdefault}{\mddefault}{\updefault}{\color[rgb]{0,0,0}$-1$}%
}}}}
\put(5926,-2086){\makebox(0,0)[lb]{\smash{{\SetFigFont{7}{8.4}{\rmdefault}{\mddefault}{\updefault}{\color[rgb]{0,0,0}$1$}%
}}}}
\put(9526,-2086){\makebox(0,0)[lb]{\smash{{\SetFigFont{7}{8.4}{\rmdefault}{\mddefault}{\updefault}{\color[rgb]{0,0,0}$4=\ell$}%
}}}}
\end{picture}%
\caption{Dimerization in a chain of length $\ell=4$.}
\label{fig dimers}
\efig
\end{centering}

Our proof of dimerization is based on the random loop representation of \cite{AN}. We introduce excitation contours in a background of 
dimerized short loops, a setting that allows to use a Peierls argument. It is presented in Section \ref{sec contour rep}. More precisely, 
Theorem \ref{thm dimers} is an immediate consequence of Proposition \ref{prop Peierls} combined with Proposition \ref{prop AN}, 
Lemma \ref{lem simple configs}, and Proposition \ref{prop surround}.

The ground state is no longer translation invariant, but it still has spin rotation invariance, as shown in the following theorem. In particular, there is 
no N\'eel order. In fact, this theorem implies that in the two periodic ground states all correlations (not just spin-spin correlations) decay 
exponentially, and therefore are extremal periodic ground states. This supports our conjecture that they are they only ground states of the 
infinite chain.

\begin{theorem}\label{thm exp decay}
For $S\geq8$, there exist $C,\eta>0$ such that
\[
\bigl| \langle S_x^i S_y^j \rangle_{\infty,\ell} \bigr| \leq C \e{-|x-y|/\eta},
\]
for all $i,j = 1,2,3$, all $x,y \in \{-\ell+1,\dots,\ell\}$, and all $\ell\in\bbN$.
\end{theorem}

The proof of this theorem can be found at the end of Section \ref{sec contour rep}; it is based on the same properties of the random loop representation that we prove for Theorem \ref{thm dimers}.
Correlations between $x$ and $y$ can be expressed in terms of events in which loops (contours) connect $x$ to $y$, and which require
that these points are surrounded by contours of size larger than $|x-y|$. The probability of these contours decays exponentially fast with 
respect to their size. Because {\em all} correlation functions can be expressed in terms of loop connectivity (see \cite{AN}), it follows
that all correlations decay exponentially.
Notice that Theorem \ref{thm exp decay} holds for smaller $S$ than Theorem \ref{thm dimers}; the reason is that it only involves large loops.

Theorem \ref{thm exp decay} also implies that the translation symmetry is spontaneously broken in the ground state for all $S\geq 8$. This follows from
\cite[Theorem 6.1]{AN}. For $S\geq 40$, Theorem \ref{thm dimers} gives a quantitative estimate of this non-translation invariance in terms of the
probability that two nearest neighbor spins form a dimer. 

\section{Loop representation and contours}
\label{sec loop rep}

We now give a description of the random loop representation of the Gibbs states defined in \eq{gibbs}.
It is convenient to restrict to $\beta\in\bbN$.

\subsection{Loops}

Let $E_\ell = \bigl\{ \{x,x+1\} : -\ell+1 \leq x \leq \ell-1 \bigr\}$ denote the set of edges. We also introduce the set $E_\ell^1$ of edges that are expected 
to preferentially host dimers (those that are shaded in Fig.\ \ref{fig dimers})
\be
E_\ell^1 = \bigl\{ \{x,x+1\} : x \in \{-\ell+1,-\ell+3,\dots,\ell-1\} \bigr\},
\ee
and $E_\ell^2 = E_\ell \setminus E_\ell^1$.
The length of the chain is chosen to be even and this, it turns out, makes it more likely that we have dimers at both edges of the 
chain. This is the mechanism that will select one of the two ground states.

In contrast to previous applications of the random loop representation \cite{AN,Uel}, it will be convenient for us here to use a discrete version of the 
loop representation, defined as follows. Let $n \in \bbN$, and consider the set $T_{\beta,n}$ of discrete times, 
$T_{\beta,n} = \frac1n \bbZ \cap [-\beta, \beta]$. A configuration $\omega$ is a subset of $E_\ell \times T_{\beta,n}$; we say that a ``double 
bar" is present at $\{x,x+1\} \times t \in E_\ell \times T_{\beta,n}$ whenever $\{x,x+1\} \times t \in \omega$; there are no double bars at $\{x,x+1\} \times t$ 
otherwise. We let $\Omega_{\ell,n}$ denote the set of configurations where no more than one double bar occurs at any given time; it is also useful to exclude double bars at time 0.

To a configuration $\omega \in \Omega_{\ell,n}$ corresponds a set of {\bf loops}, as illustrated in Fig.\ \ref{fig loops}. A loop $\gamma$ is a closed trajectory
\be
\begin{split}
[0,L]_{\rm per} &\to \{-\ell+1,\dots,\ell\} \times [-\beta,\beta]_{\rm per} \\
t &\mapsto \gamma(t) = \bigl( x(t), T(t) \bigr)
\end{split}
\ee
such that $x(t)$ is piecewise constant and $T'(t) = \pm1$. We let $L(\gamma)$ denote the length of $\gamma$, that is, the smallest $L>0$ in the above equation. A jump occurs at $t \in [0,L(\gamma)]$ provided $\{ x(t-), x(t+) \} \times T(t)$ contains a double bar, and that $T'(t-) = - T'(t+)$. We identify loops with identical support. Let $\caL(\omega)$ denote the number of loops of the configuration $\omega$. We always have $1 \leq \caL(\omega) \leq 2\ell + |\omega|$ and
\be
\sum_{j=1}^{\caL(\omega)} L(\gamma_j) = 2\beta \ell.
\ee

\begin{centering}
\bfig
\begin{picture}(0,0)%
\includegraphics{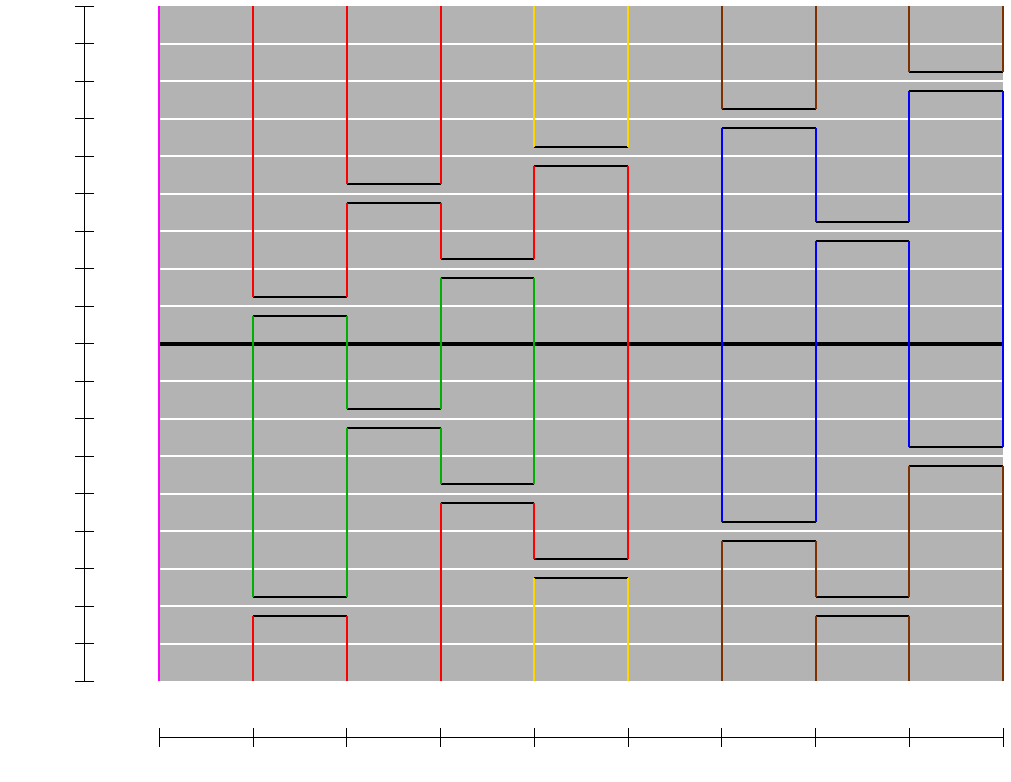}%
\end{picture}%
\setlength{\unitlength}{2368sp}%
\begingroup\makeatletter\ifx\SetFigFont\undefined%
\gdef\SetFigFont#1#2#3#4#5{%
  \reset@font\fontsize{#1}{#2pt}%
  \fontfamily{#3}\fontseries{#4}\fontshape{#5}%
  \selectfont}%
\fi\endgroup%
\begin{picture}(8018,6272)(1166,-6560)
\put(2101,-6511){\makebox(0,0)[lb]{\smash{{\SetFigFont{7}{8.4}{\rmdefault}{\mddefault}{\updefault}{\color[rgb]{0,0,0}$-\ell+1$}%
}}}}
\put(9076,-6511){\makebox(0,0)[lb]{\smash{{\SetFigFont{7}{8.4}{\rmdefault}{\mddefault}{\updefault}{\color[rgb]{0,0,0}$\ell$}%
}}}}
\put(5326,-6511){\makebox(0,0)[lb]{\smash{{\SetFigFont{7}{8.4}{\rmdefault}{\mddefault}{\updefault}{\color[rgb]{0,0,0}$0$}%
}}}}
\put(1336,-2801){\makebox(0,0)[lb]{\smash{{\SetFigFont{6}{7.2}{\rmdefault}{\mddefault}{\updefault}{\color[rgb]{0,0,0}$1/n$}%
}}}}
\put(1491,-3096){\makebox(0,0)[lb]{\smash{{\SetFigFont{6}{7.2}{\rmdefault}{\mddefault}{\updefault}{\color[rgb]{0,0,0}$0$}%
}}}}
\put(1416,-5806){\makebox(0,0)[lb]{\smash{{\SetFigFont{6}{7.2}{\rmdefault}{\mddefault}{\updefault}{\color[rgb]{0,0,0}$-\beta$}%
}}}}
\put(1546,-396){\makebox(0,0)[lb]{\smash{{\SetFigFont{6}{7.2}{\rmdefault}{\mddefault}{\updefault}{\color[rgb]{0,0,0}$\beta$}%
}}}}
\put(1331,-2491){\makebox(0,0)[lb]{\smash{{\SetFigFont{6}{7.2}{\rmdefault}{\mddefault}{\updefault}{\color[rgb]{0,0,0}$2/n$}%
}}}}
\put(1186,-3396){\makebox(0,0)[lb]{\smash{{\SetFigFont{6}{7.2}{\rmdefault}{\mddefault}{\updefault}{\color[rgb]{0,0,0}$-1/n$}%
}}}}
\put(1166,-3691){\makebox(0,0)[lb]{\smash{{\SetFigFont{6}{7.2}{\rmdefault}{\mddefault}{\updefault}{\color[rgb]{0,0,0}$-2/n$}%
}}}}
\end{picture}%
\caption{Loop representation of the $SU(2S+1)$-invariant quantum spin chains.}
\label{fig loops}
\efig
\end{centering}

The relevant probability measure on $\Omega_{\ell,n}$ involves the number of loops and is given by
\be
\mu_{\beta,\ell,n}(\omega) = \frac1{Z_n(\beta,\ell)} \bigl( \tfrac1n \bigr)^{|\omega|} (2S+1)^{\caL(\omega)-|\omega|},
\ee
with
\be
Z_n(\beta,\ell) = \sum_{\omega \in \Omega_{\ell,n}} \bigl( \tfrac1n \bigr)^{|\omega|} (2S+1)^{\caL(\omega)-|\omega|}.
\ee
In the following, we use the notation $\bbP_{\beta,\ell,n}$ and $ \bbE_{\beta,\ell,n}$ for the probability and expectation with respect to the measure $\mu_{\beta,\ell,n}$. The event $x \leftrightarrow y$ (resp.\ $x \not\leftrightarrow y$) represents the set of all $\omega \in \Omega_{\ell,n}$ where $x \times 0$ and $y \times 0$ belong to the same loop (resp.\ belong to distinct loops).

\begin{proposition}
\label{prop AN}
\hfill

(a) $\displaystyle \Tr \e{-2\beta H_\ell} = \lim_{n\to\infty} Z_n(\beta,\ell)$.

(b) $\displaystyle \langle P_{x,x+1}^{(0)} \rangle_{2\beta,\ell} = \tfrac1{(2S+1)^2} + \bigl( 1 - \tfrac1{(2S+1)^2} \bigr) \lim_{n\to\infty} \bbP_{\beta, \ell,n}(x \leftrightarrow x+1)$.
\end{proposition}

\begin{proof}
We only prove (b), since the proof of (a) is similar and it can be found in \cite{AN}. Using  the Trotter formula, we have
\be
\begin{split}
\Tr \e{-\beta H_\ell} P_{x,x+1}^{(0)} \e{-\beta H_\ell} &= \lim_{n\to\infty} \Tr \prod_{\{y,y+1\} \times t \in E_\ell \times T_{\beta,n}} \bigl[ 1 + \tfrac1n P_{y,y+1}^{(0)} \bigr] \\
&\hspace{35mm} P_{x,x+1}^{(0)} \prod_{\{y,y+1\} \times t \in E_\ell \times T_{\beta,n}} \bigl[ 1 + \tfrac1n P_{y,y+1}^{(0)} \bigr] \\
&= \lim_{n \to \infty} \sum_{\omega \in \Omega_{\ell,n}} \bigl( \tfrac1n \bigr)^{|\omega|} \; \Tr \prod_{\{y,y+1\} \times t \,\in\, \bar\omega} P_{y,y+1}^{(0)}.
\end{split}
\ee
Here, the configuration $\bar\omega$ is equal to the configuration $\omega$ with an extra double bar at $\{x,x+1\} \times 0$, and the last product is ordered by increasing times $t \in T_{\beta,n}$. It is not hard to verify that
\be
\Tr \prod_{\{y,y+1\} \times t \,\in\, \bar\omega} P_{y,y+1}^{(0)} = (2S+1)^{\caL(\bar\omega)-|\bar\omega|}.
\ee
Next, observe that
\be
\caL(\bar\omega) = \begin{cases} \caL(\omega)+1 & \text{if } x \leftrightarrow x+1; \\ \caL(\omega)-1 & \text{if } x \not\leftrightarrow x+1. \end{cases}
\ee
As a consequence, we have
\be
\begin{split}
\Tr \e{-\beta H_\ell} P_{x,x+1}^{(0)} &\e{-\beta H_\ell} = \lim_{n\to\infty} \sum_{\omega \in \Omega_{\ell,n}}\bigl( \tfrac1n \bigr)^{|\omega|} \; (2S+1)^{\caL(\omega)-|\omega|-1} \\
&\hspace{28mm} \Bigl[ (2S+1) 1_{x \leftrightarrow x+1}(\omega) + \tfrac1{2S+1} 1_{x \not\leftrightarrow x+1}(\omega) \Bigr] \\
&= \lim_{n\to\infty} Z_n(\beta,\ell) \Bigl[ \bbP_{\beta,\ell,n}(x \leftrightarrow x+1) + \tfrac1{(2S+1)^2} \bbP_{\beta,\ell,n}(x \not\leftrightarrow x+1) \Bigr],
\end{split}
\ee
and the claim (b) of the proposition follows.
\end{proof}

\begin{centering}
\bfig
\begin{picture}(0,0)%
\includegraphics{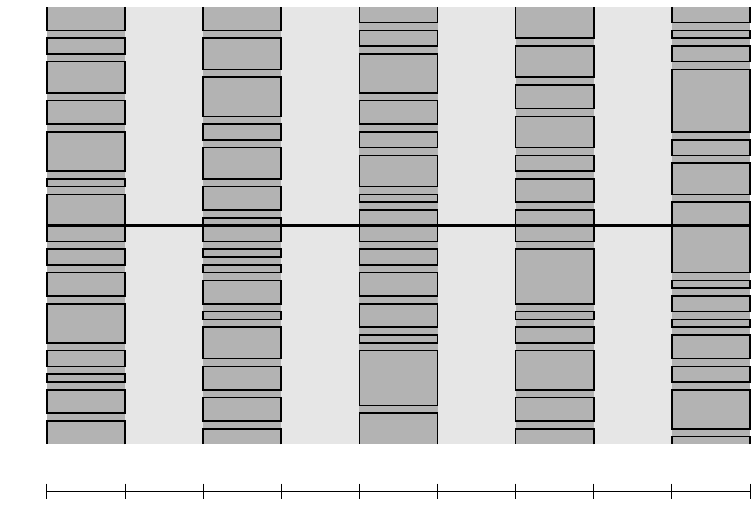}%
\end{picture}%
\setlength{\unitlength}{1973sp}%
\begingroup\makeatletter\ifx\SetFigFont\undefined%
\gdef\SetFigFont#1#2#3#4#5{%
  \reset@font\fontsize{#1}{#2pt}%
  \fontfamily{#3}\fontseries{#4}\fontshape{#5}%
  \selectfont}%
\fi\endgroup%
\begin{picture}(7233,5056)(1951,-6560)
\put(9076,-6511){\makebox(0,0)[lb]{\smash{{\SetFigFont{6}{7.2}{\rmdefault}{\mddefault}{\updefault}{\color[rgb]{0,0,0}$\ell$}%
}}}}
\put(1951,-5761){\makebox(0,0)[lb]{\smash{{\SetFigFont{6}{7.2}{\rmdefault}{\mddefault}{\updefault}{\color[rgb]{0,0,0}$-\beta$}%
}}}}
\put(5326,-6511){\makebox(0,0)[lb]{\smash{{\SetFigFont{6}{7.2}{\rmdefault}{\mddefault}{\updefault}{\color[rgb]{0,0,0}$0$}%
}}}}
\put(2026,-6511){\makebox(0,0)[lb]{\smash{{\SetFigFont{6}{7.2}{\rmdefault}{\mddefault}{\updefault}{\color[rgb]{0,0,0}$-\ell+1$}%
}}}}
\put(2101,-3736){\makebox(0,0)[lb]{\smash{{\SetFigFont{6}{7.2}{\rmdefault}{\mddefault}{\updefault}{\color[rgb]{0,0,0}$0$}%
}}}}
\put(2101,-1636){\makebox(0,0)[lb]{\smash{{\SetFigFont{6}{7.2}{\rmdefault}{\mddefault}{\updefault}{\color[rgb]{0,0,0}$\beta$}%
}}}}
\end{picture}%
\caption{Illustration for `optimal' configurations that allow for large numbers of loops.}
\label{fig background}
\efig
\end{centering}

The measure $\mu_{\beta,\ell,n}$ is biased towards configurations with many loops. The optimal way to stack many loops is to have only double bars on edges $\{x,x+1\} \in E_\ell^1$, see Fig.\ \ref{fig background}. If $\omega$ has this form, we have
\be
1_{x \leftrightarrow x+1}(\omega) = \begin{cases} 1 & \text{if } \{x,x+1\} \in E_\ell^1, \\ 0 & \text{if } \{x,x+1\} \in E_\ell^2. \end{cases}
\ee
If these were typical configurations, we would get $\bbP_{\beta, \ell,n}(x \leftrightarrow x+1) = 1 > 0 = \bbP_{\beta,\ell,n}(x \not\leftrightarrow x+1)$ for all 
$x \in \{-\ell+1, -\ell+3, \dots, \ell-1\}$, and Theorem \ref{thm dimers} follows from Proposition \ref{prop AN}. Actual typical configurations are naturally 
more complicated but, for large $S$, not very much, as they possess this structure up to sparse `excitations'. This is shown below. We put excitations
in quotes here because they are present with a non-zero density in the ground states, and can be regarded as representing the quantum fluctuations 
at zero temperature.

\subsection{Simplified set of configurations without winding loops}

As is usual, we will derive our estimates for finite systems, which in our case means finite chains at finite inverse temperature $\beta$. The estimates 
will then carry over to the limit of infinite $\beta$ and the infinite chain. In the limit $\beta\to\infty$, the so-called {\em winding loops} will have vanishing
probability. These winding loops are a complication for the Peierls-type argument we want to develop. Therefore, it will be helpful to work with a 
restricted set of configurations in which almost all winding loops have been eliminated. To do this we need to show that the probability of the 
configurations we ignore indeed vanishes in the limit $\beta\to\infty$. This is the purpose of the next lemma.

Given an integer $\alpha<\beta$, let $\Omega_{\ell,n}^\alpha$ denote the set of configurations where intervals $\{x,x+1\} \times [\alpha,\alpha+1]$ contain at least one double bar if $\{x,x+1\} \in E_\ell^1$, and none if $\{x,x+1\} \in E_\ell^2$. These configurations possess a convenient, spontaneous boundary condition in the time direction (this is depicted in Fig.\ \ref{fig contours}).
Almost all configurations have this property for some $\alpha \in \bbN$:

\begin{lemma}
\label{lem simple configs}
We have
\[
\lim_{\beta\to\infty} \bbP_{\beta,\ell,n} \bigl( \cup_{\alpha=1}^\beta \Omega_{\ell,n}^\alpha \bigr) = 1.
\]
The limit $\beta\to\infty$ converges uniformly in $n$.
\end{lemma}

\begin{proof}
Let $A_\alpha = \cup_{\alpha'<\alpha} \Omega_{\ell,n}^{\alpha'}$. We have $A_{\alpha+1} \supset A_\alpha$ for all $\alpha \in \bbN$, and one checks below that
\be
\label{il faut verifier}
\bbP_{\beta,\ell,n}(A_{\alpha+1} | A_\alpha^{\rm c}) > c,
\ee
with $c>0$ independent of $\alpha, \beta, n$ (it depends on $S,\ell$, though). Then
\be
\bbP_{\beta,\ell,n}(A_\beta^{\rm c}) = \bbP_{\beta,\ell,n}(A_\beta^{\rm c} | A_{\beta-1}^{\rm c}) \, \bbP_{\beta,\ell,n}(A_{\beta-1}^{\rm c}) \leq (1-c) \, \bbP_{\beta,\ell,n}(A_{\beta-1}^{\rm c}).
\ee
Iterating, we find that $\bbP_{\beta,\ell,n}(A_\beta^{\rm c})$ is less than $(1-c)^\beta$, which implies the lemma.

In order to check \eqref{il faut verifier}, let $\hat A_\alpha \subset A_\alpha^{\rm c}$ be the set of configurations such that no double bars occur at times in $[\alpha,\alpha+1] \cap T_{\beta,n}$. A configuration $\omega \in A_\alpha^{\rm c}$ can be decomposed as $\omega = \hat\omega \cup \omega'$, where $\hat\omega \in \hat A_\alpha$ and $\omega'$ contains only double bars at times in $[\alpha,\alpha+1] \cap T_{\beta,n}$. We have
\be
\caL(\hat\omega) - |\omega'| \leq \caL(\omega) \leq \caL(\hat\omega) + |\omega'|,
\ee
so that
\be
\mu_{\beta,\ell,n}(\hat\omega) \bigl(\tfrac1n\bigr)^{|\omega'|} (2S+1)^{-2|\omega'|} \leq \mu_{\beta,\ell,n}(\omega) \leq \mu_{\beta,\ell,n}(\hat\omega) \bigl(\tfrac1n\bigr)^{|\omega'|}.
\ee
Then
\be
\begin{split}
\bbP_{\beta,\ell,n}(A_{\alpha+1} | A_\alpha) &= \frac{\bbP_{\beta,\ell,n}( \Omega_{\ell,n}^\alpha \cap A_\alpha^{\rm c})}{\bbP_{\beta,\ell,n}(A_\alpha^{\rm c})} = \frac{\sum_{\hat\omega \in \hat A_\alpha} \sum_{\omega' : \Omega_{\ell,n}^\alpha} \mu_{\beta,\ell,n}(\hat\omega \cup \omega')}{\sum_{\hat\omega \in \hat A_\alpha} \sum_{\omega'} \mu_{\beta,\ell,n}(\hat\omega \cup \omega')} \\
&\geq \frac{\bigl( 1 + \frac1n (2S+1)^{-2} \bigr)^{\ell n}}{\bigl( 1 + \frac1n \bigr)^{(2\ell+1) n}}.
\end{split}
\ee
The last term is indeed bounded away from 0, uniformly in $\alpha,\beta,n$.
\end{proof}

Loops visiting only one or two sites are called {\em short loops}. The loops shown in Fig.\ \ref{fig background} are all short loops. 
We say that a loop is {\em long} if it is not short.
Since the loops of configurations of $\Omega_{\ell,n}^\alpha$ do not wind, they have an interior in the sense of Jordan curves; we actually call ``interior" its intersection with $\{-\ell+1,\dots,\ell\} \times T_{\beta,n}$.
For $\omega \in \Omega_{\ell,n}^\alpha$, we introduce the event $E_x^\circlearrowleft$ 
where $(x,0)$ belongs to a long loop, or to the interior of a long loop. See Eq.\ \eqref{def surround} for an equivalent definition that involves contours, to be defined below.

\begin{proposition}
\label{prop surround}
If $x \in \{ -\ell+3, -\ell+5, \dots, \ell-1 \}$, we have
\[
\bbP_{\beta,\ell,n}(x \leftrightarrow x+1 | \Omega_{\ell,n}^\alpha) - \bbP_{\beta,\ell,n}(x-1 \leftrightarrow x | \Omega_{\ell,n}^\alpha) \geq 1 - 2 \bbP_{\beta,\ell,n}(E_x^\circlearrowleft | \Omega_{\ell,n}^\alpha).
\]
\end{proposition}

\begin{proof}
Let $\omega \in \Omega_{\ell,n}^\alpha$. If $\omega \notin E_x^\circlearrowleft$, then $(x,0)$ belongs to the exterior of all long loops and we have $x \leftrightarrow x+1$ and $x-1 \not\leftrightarrow x$. Then
\be
\begin{split}
&\bbP_{\beta,\ell,n}(x \leftrightarrow x+1 | \Omega_{\ell,n}^\alpha) \geq 1 - \bbP_{\beta,\ell,n}(E_x^\circlearrowleft | \Omega_{\ell,n}^\alpha), \\
&\bbP_{\beta,\ell,n}(x-1 \not\leftrightarrow x | \Omega_{\ell,n}^\alpha) \geq 1 - \bbP_{\beta,\ell,n}(E_x^\circlearrowleft | \Omega_{\ell,n}^\alpha).
\end{split}
\ee
The claim of the proposition follows immediately.
\end{proof}

There remains to show that $\bbP_{\beta,\ell,n}(E_x^\circlearrowleft | \Omega_{\ell,n}^\alpha)$ is less than $\frac12$, uniformly in $\beta,\ell,n$.

\subsection{Blocks and domains}

We define a ``block" to be a set $\{x,x+1\} \times I$, where $\{x,x+1\} \in E_\ell^1$ and $I$ is a proper interval in $T_{\beta,n}$. A ``domain" $D$ is a finite collection of disjoint blocks.

Let $\Omega_D$ denote the set of configurations of double bars on $D$; that is, double bars within blocks of $D$, or between blocks involving nearest neighbors at equal times. We let $Z(D)$ denote the partition function on $D$, namely,
\be
Z(D) = \sum_{\omega \in \Omega_D} \bigl( \tfrac1n \bigr)^{|\omega|} (2S+1)^{\caL_D(\omega)- |\omega|}.
\ee
Loops are defined as before, with the understanding that bars are present at the ends of each block of $D$ and $\caL_D(\omega)$ is the number of loops. Let $\bbP_D$ denote the probability measure on configurations on $D$.
Let
\be
D_\alpha = \Union_{\{x,x+1\} \in E_\ell^1} \{x,x+1\} \times [\alpha+1,\alpha],
\ee
where $[\alpha+1,\alpha]$ is the interval in $T_{\beta,n}$ that contains $0$.

\begin{lemma}
\label{lem contour expression}
We have
\[
\bbP_{\beta,\ell,n}(E_x^\circlearrowleft | \Omega_{\ell,n}^\alpha) = \bbP_{D_\alpha}(E_x^\circlearrowleft).
\]
\end{lemma}

\begin{proof}
In the left side, the sum over $\omega \in \Omega_{\ell,n}^\alpha$ can be done by summing over $\omega \in \Omega_{D_\alpha}$, and by summing over configurations $\omega'$ of double bars on $E_\ell^1 \times [\alpha,\alpha+1]$. We have
\be
\sum_{\omega \in E_x^\circlearrowleft \cap \Omega_{\ell,n}^\alpha} \bigl( \tfrac1n \bigr)^{|\omega|} (2S+1)^{\caL(\omega) - |\omega|} = \sum_{\omega \in \Omega_{D_\alpha}} \sum_{\omega'} 1_{E_x^\circlearrowleft}(\omega \cup \omega') \bigl( \tfrac1n \bigr)^{|\omega \cup \omega'|} (2S+1)^{\caL(\omega \cup \omega') - |\omega \cup \omega'|}.
\ee
We have $|\omega \cup \omega'| = |\omega| + |\omega'|$ and
\be
\caL(\omega \cup \omega') = \caL_{D_\alpha}(\omega) + |\omega'| - \ell.
\ee
Finally, observe that the event $E_x^\circlearrowleft$ depends on $\omega$ only; the contribution of $\omega'$ can be factored out and the lemma follows.
\end{proof}

\section{Contour representation}\label{sec contour rep}

We now introduce a contour representation suitable for executing a Peierls argument. Developing a Peierls argument for our model is not entirely straightforward. In particular,  since the cost of large contours is entropic and we need to estimate the sparsity of loops instead of relying on an energy estimate which is usually quite immediate. We found it necessary to condition on arbitrary configurations of external contours not involving a given point $(x,0)$. This allows to apply a simpler Peierls argument with a single contour, but in domains of arbitrary shape. We begin with a number of definitions.

\begin{centering}
\bfig
\includegraphics[width=130mm]{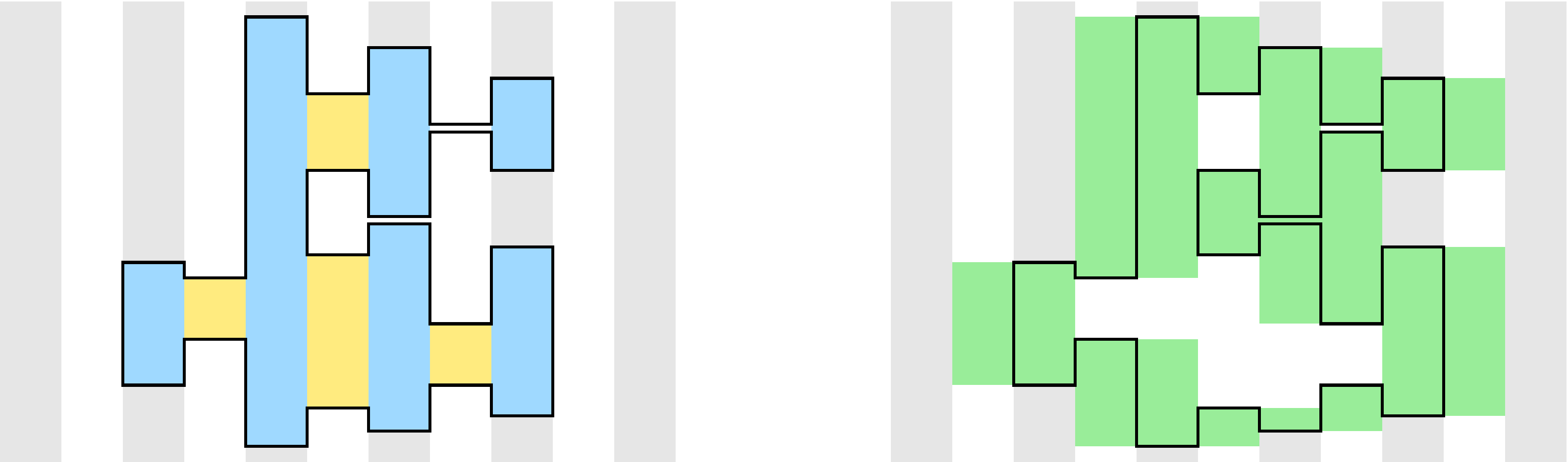}
\caption{(Color online) Left: a contour $\gamma$ with its 1-interior ${\rm int}_1\gamma$ (blue) and its 2-interior ${\rm int}_2\gamma$ (yellow). Right: the same contour with its support ${\rm supp}\,\gamma$ (green).}
\label{fig one contour}
\efig
\end{centering}

An edge in space-time is of the form $(\{x,x+1\}, t)$, $x\in [-\ell, \ell-1]\subset \mathbb{Z}, t\in [-\beta,\beta]\subset \mathbb{R}$. In figures it is more
convenient to replace $\{x,x+1\}$ by $[x,x+1]\subset\Rl$. In the same way, in figures we will depict a block as a rectangular region in 
$\mathbb{R}^2$ of the form $[x,x+1]\times [a,b]$.
 
We now call  {\em contour} a long loop; a contour $\gamma$ is characterized by
\begin{itemize}
\item its number of double bars, $|\gamma|$;
\item its 1-interior ${\rm int}_1\gamma$, that is, the set of edges $\{x,x+1\} \times t$ inside the enclosed area, with $\{x,x+1\} \in E_\ell^1$;
\item its 2-interior ${\rm int}_2\gamma$, that is, the set of edges $\{x,x+1\} \times t$ inside the enclosed area, with $\{x,x+1\} \in E_\ell^2$;
\item the length $L(\gamma)$ such that $|{\rm int}_1\gamma| - |{\rm int}_2\gamma| = \tfrac12 n L(\gamma)$, where $|\cdot|$ means cardinality. $L(\gamma)$ is almost equal to the lengths of vertical legs of the loop (it is equal up to $\frac{|\gamma|}n$), because each stretch of the boundary of ${\rm int}_1\gamma$ belongs either to the perimeter, or to the vertical boundary of ${\rm int}_2\gamma$;
\item its support ${\rm supp}\,\gamma$, the set of edges of $E_\ell$ with at least one endpoint on the loop;
\item its exterior ${\rm ext}\,\gamma$, which is equal to the union of blocks outside the enclosed area; 
\end{itemize}

A contour with its interiors and its support is displayed on Fig.\ \ref{fig one contour}.
Notice that the event $E_x^\circlearrowleft$ is equivalent to
\be
\label{def surround}
E_x^\circlearrowleft = \bigl\{ \omega \in \Omega_{\ell,n}^\alpha : \text{ $\exists$ contour $\gamma$ such that } (x,0) \in {\rm int}_1\gamma \bigr\}.
\ee
Given $\omega \in \Omega_{\ell,n}^\alpha$, we call ``external contour" a loop such that
\begin{itemize}
\item it involves at least a jump $\{x,x+1\} \in E_\ell^2$;
\item it is not surrounded by another loop.
\end{itemize}
This definition is illustrated in Fig.\ \ref{fig contours}. We identify an external contour $\gamma$ with the set of double bars that it involves; that is, $\gamma$ denotes both a closed trajectory, and an element of $\Omega_{D_\alpha}$.

\begin{centering}
\bfig
\begin{picture}(0,0)%
\includegraphics{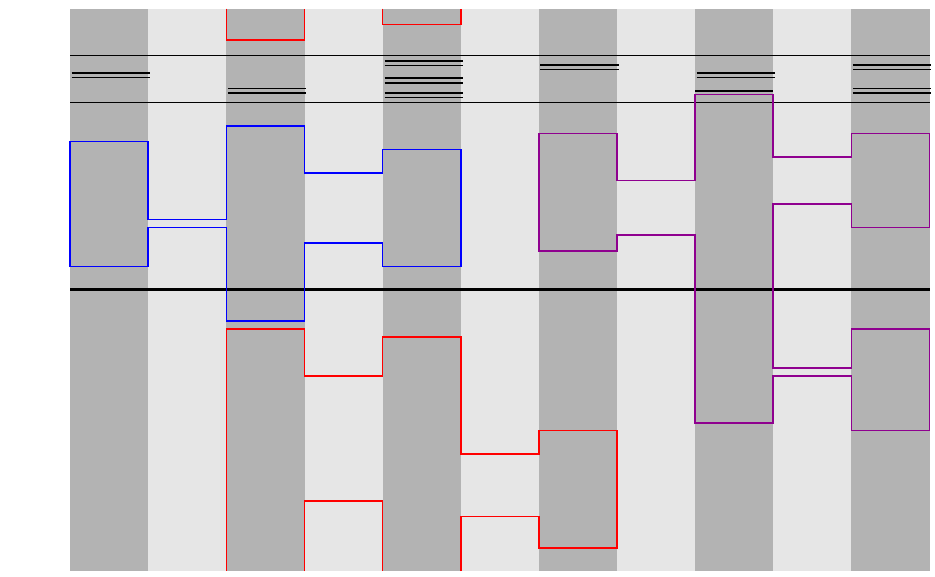}%
\end{picture}%
\setlength{\unitlength}{1973sp}%
\begingroup\makeatletter\ifx\SetFigFont\undefined%
\gdef\SetFigFont#1#2#3#4#5{%
  \reset@font\fontsize{#1}{#2pt}%
  \fontfamily{#3}\fontseries{#4}\fontshape{#5}%
  \selectfont}%
\fi\endgroup%
\begin{picture}(8962,5506)(226,-5810)
\put(451,-5761){\makebox(0,0)[lb]{\smash{{\SetFigFont{6}{7.2}{\rmdefault}{\mddefault}{\updefault}{\color[rgb]{0,0,0}$-\beta$}%
}}}}
\put(601,-3136){\makebox(0,0)[lb]{\smash{{\SetFigFont{6}{7.2}{\rmdefault}{\mddefault}{\updefault}{\color[rgb]{0,0,0}$0$}%
}}}}
\put(601,-436){\makebox(0,0)[lb]{\smash{{\SetFigFont{6}{7.2}{\rmdefault}{\mddefault}{\updefault}{\color[rgb]{0,0,0}$\beta$}%
}}}}
\put(226,-886){\makebox(0,0)[lb]{\smash{{\SetFigFont{6}{7.2}{\rmdefault}{\mddefault}{\updefault}{\color[rgb]{0,0,0}$\alpha+1$}%
}}}}
\put(601,-1336){\makebox(0,0)[lb]{\smash{{\SetFigFont{6}{7.2}{\rmdefault}{\mddefault}{\updefault}{\color[rgb]{0,0,0}$\alpha$}%
}}}}
\end{picture}%
\caption{A configuration of $\Omega_{\ell,n}^\alpha$ with three external contours.}
\label{fig contours}
\efig
\end{centering}

A difficulty with our Peierls argument is to control the change in the number of loops when
erasing a contour. We find it convenient to condition on the configuration of external contours
away from (x,0); a similar idea was used in \cite{GS}. This gives a domain with
complicated shape, but this is not a problem. The background configuration of loops is now simpler
and this is useful.
Given $\omega \in \Omega_{D_\alpha}$, consider the set $\Gamma$ of external contours that do not surround $(x,0)$. Let ${\rm ext}\,\Gamma = D_\alpha \setminus {\rm int}_1 \Gamma$, and let $X(\omega) \subset {\rm ext}\,\Gamma$ denote the connected component that contains $(x,0)$. That is, the domain $X(\omega)$ is a subset of $D_\alpha$, and the graph whose vertices are the blocks, and with edges between blocks that can be connected by a double-bar, is connected. This is illustrated in Fig.\ \ref{fig set X}. Finally, let $\widetilde\Omega_X \subset \Omega_X$ denote the configurations without contours, or with just one external contour that surrounds $(x,0)$; and let $\widetilde\bbP_X(\cdot) = \bbP_X(\cdot | \widetilde\Omega_X)$ denote the conditional probability.

\begin{lemma}
\label{lem special Peierls}
There exists $p_X \geq 0$ such that $\sum_X p_X = 1$ and
\[
\bbP_{D_\alpha}(E_x^\circlearrowleft) = \sum_X p_X \widetilde\bbP_X(E_x^\circlearrowleft).
\]
\end{lemma}

\begin{proof}
We have
\be
\bbP_{D_\alpha}(E_x^\circlearrowleft) = \frac1{Z(D_\alpha)} \sum_X \sumtwo{\omega: X(\omega)=X}{\phantom{\omega:} \omega|_X = \emptyset} \sum_{\omega' \in \widetilde\Omega_X \cap E_x^\circlearrowleft} \bigl( \tfrac1n \bigr)^{|\omega \cup \omega'|} (2S+1)^{\caL_{D_\alpha}(\omega \cup \omega') - |\omega \cup \omega'|}.
\ee
Let $b(X)$ the number of blocks of $X$; the number of loops of $\omega \cup \omega'$ satisfies the important relation
\be
\caL_{D_\alpha}(\omega \cup \omega') = \caL_{D_\alpha}(\omega) + \caL_X(\omega') - b(X).
\ee
Indeed, this holds for $\omega' = \emptyset$ since $\caL_X(\emptyset) = b(X)$; adding a double bar to $\omega'$ results in the same change in $\caL_{D_\alpha}(\omega \cup \omega')$ and $\caL_X(\omega')$.

Let $\widetilde Z(X)$ be the partition function on $\widetilde\Omega_X$; we have
\bm
\bbP_{D_\alpha}(E_x^\circlearrowleft) = \sum_X \frac1{Z(D_\alpha)} \sumtwo{\omega: X(\omega)=X}{\phantom{\omega:} \omega|_X = \emptyset} \bigl( \tfrac1n \bigr)^{|\omega|} (2S+1)^{\caL_{D_\alpha}(\omega) - |\omega| - b(X)} \widetilde Z(X) \\
\frac1{\widetilde Z(X)} \sum_{\omega' \in \widetilde\Omega_X \cap E_x^\circlearrowleft} \bigl( \tfrac1n \bigr)^{|\omega'|} (2S+1)^{\caL_X(\omega') - |\omega'|}.
\end{multline}
For given $X$, the first line of the right side gives $p_X$; the second line is $\widetilde\bbP_X(E_x^\circlearrowleft)$ and we get the lemma.
\end{proof}

\begin{centering}
\bfig
\begin{picture}(0,0)%
\includegraphics{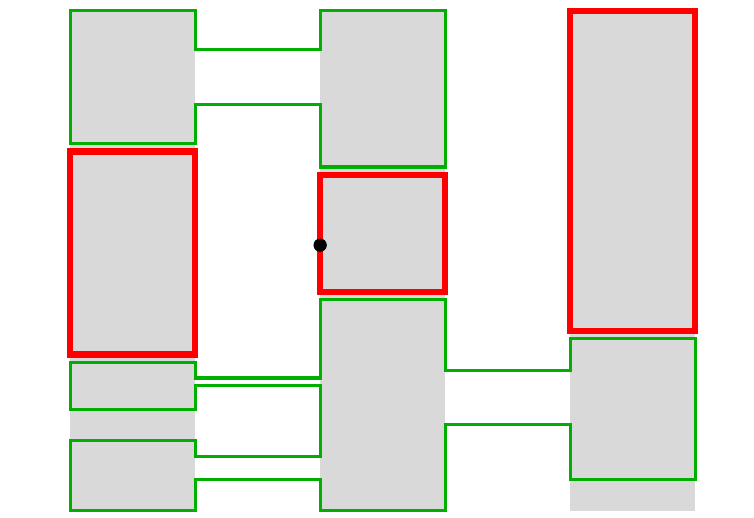}%
\end{picture}%
\setlength{\unitlength}{1973sp}%
\begingroup\makeatletter\ifx\SetFigFont\undefined%
\gdef\SetFigFont#1#2#3#4#5{%
  \reset@font\fontsize{#1}{#2pt}%
  \fontfamily{#3}\fontseries{#4}\fontshape{#5}%
  \selectfont}%
\fi\endgroup%
\begin{picture}(7211,4956)(526,-5224)
\put(901,-436){\makebox(0,0)[lb]{\smash{{\SetFigFont{7}{8.4}{\rmdefault}{\mddefault}{\updefault}{\color[rgb]{0,0,0}$\alpha$}%
}}}}
\put(7276,-1636){\makebox(0,0)[lb]{\smash{{\SetFigFont{9}{10.8}{\rmdefault}{\mddefault}{\updefault}{\color[rgb]{0,0,0}$X$}%
}}}}
\put(7276,-5086){\makebox(0,0)[lb]{\smash{{\SetFigFont{9}{10.8}{\rmdefault}{\mddefault}{\updefault}{\color[rgb]{0,0,0}$D_\alpha$}%
}}}}
\put(526,-5161){\makebox(0,0)[lb]{\smash{{\SetFigFont{7}{8.4}{\rmdefault}{\mddefault}{\updefault}{\color[rgb]{0,0,0}$\alpha+1$}%
}}}}
\put(2926,-2536){\makebox(0,0)[lb]{\smash{{\SetFigFont{7}{8.4}{\rmdefault}{\mddefault}{\updefault}{\color[rgb]{0,0,0}$(x,0)$}%
}}}}
\end{picture}%
\caption{The set $X$ around $(x,0)$.}
\label{fig set X}
\efig
\end{centering}

In view of Proposition \ref{prop surround}, to prove Theorem \ref{thm dimers} we need to show that, if $S\geq 40$, we 
have $\widetilde\bbP_X(E_x^\circlearrowleft) < \frac12$ for all $X$. This will be achieved in Proposition \ref{prop Peierls}.
In preparation for that estimate, we need two lemmas.

If $\widetilde\omega \in \Omega_{X}$ is a configuration that contains the contour $\gamma$, then it is the disjoint union
\be
\widetilde\omega = \gamma \cup \omega \cup \omega',
\ee
where $\omega \in \Omega_{{\rm ext}\, \gamma}$ and $\omega' \in \Omega_{{\rm int}_2 \gamma}$. Let ${\rm rint}_2 \gamma $ denote the shift of ${\rm int}_2 \gamma$ by one unit to the right; if $\omega' \in \Omega_{{\rm int}_2 \gamma}$, we also let ${\rm r}\omega' \in \Omega_{{\rm int}_2 \gamma}$ denote the right-shift of $\omega'$ by one unit. The next lemma is stated and proved for any configurations of $\Omega_X$, but we only need it for configurations of $\widetilde\Omega_X$.

\begin{lemma}
\label{lem last difficulty}
Let $\gamma$ be an external contour, $\omega \in \Omega_{{\rm ext}\, \gamma}$, and $\omega' \in \Omega_{{\rm int}_2 \gamma}$. Then the configuration $\gamma \cup \omega \cup \omega' \in \Omega_{X}$ satisfies
\[
\caL_X(\gamma \cup \omega \cup \omega') \leq \caL_X(\omega \cup {\rm r}\omega') + \tfrac12 |\gamma| + 1.
\]
\end{lemma}

\begin{proof}
See Fig.\ \ref{fig picproof}; we have
\be
\caL_X(\gamma \cup \omega \cup \omega') = 1 + \caL_{{\rm ext}\,\gamma \cup {\rm rint}_2 \gamma}(\omega \cup {\rm r}\omega').
\ee
We call ``green blocks" the connected components of ${\rm int}_1 \gamma \setminus {\rm rint}_2 \gamma$. We add these blocks one by one to the domain; at each step, the number of loops may decrease by one, so that $\caL_{{\rm ext}\,\gamma \cup {\rm rint}_2 \gamma}(\omega \cup {\rm r}\omega') \leq \caL_X(\omega \cup {\rm r}\omega') + \#\text{green blocks}$. The number of green blocks is less than $\frac12 |\gamma|$.

\begin{centering}
\bfig
\begin{picture}(0,0)%
\includegraphics{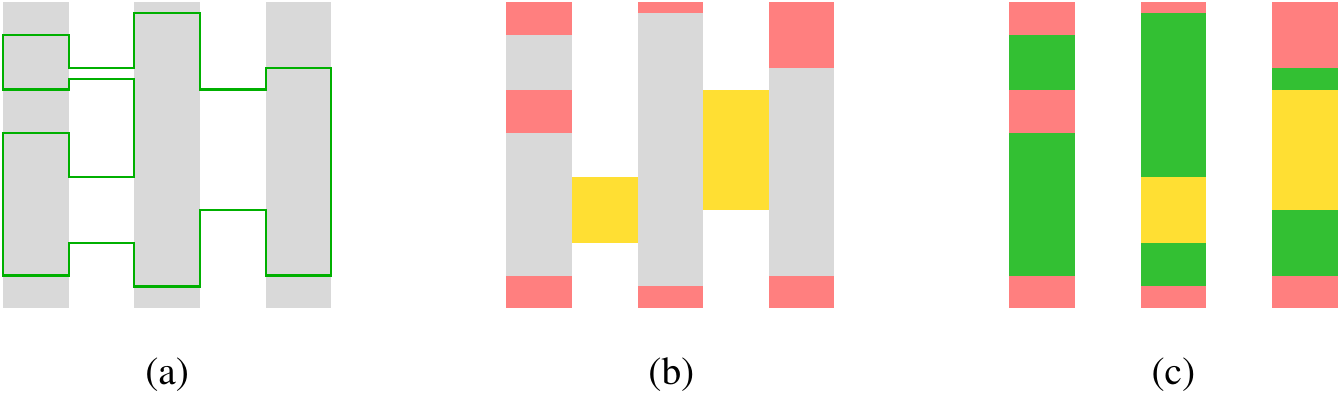}%
\end{picture}%
\setlength{\unitlength}{2763sp}%
\begingroup\makeatletter\ifx\SetFigFont\undefined%
\gdef\SetFigFont#1#2#3#4#5{%
  \reset@font\fontsize{#1}{#2pt}%
  \fontfamily{#3}\fontseries{#4}\fontshape{#5}%
  \selectfont}%
\fi\endgroup%
\begin{picture}(9173,2702)(579,-3062)
\put(2926,-1186){\makebox(0,0)[lb]{\smash{{\SetFigFont{10}{12.0}{\rmdefault}{\mddefault}{\updefault}{\color[rgb]{0,.69,0}$\gamma$}%
}}}}
\put(6376,-661){\makebox(0,0)[lb]{\smash{{\SetFigFont{10}{12.0}{\rmdefault}{\mddefault}{\updefault}{\color[rgb]{1,0,0}${\rm ext}\, \gamma$}%
}}}}
\put(4501,-1486){\makebox(0,0)[lb]{\smash{{\SetFigFont{8}{9.6}{\rmdefault}{\mddefault}{\updefault}{\color[rgb]{0,0,0}${\rm int}_2 \gamma$}%
}}}}
\put(8851,-1786){\makebox(0,0)[lb]{\smash{{\SetFigFont{8}{9.6}{\rmdefault}{\mddefault}{\updefault}{\color[rgb]{0,0,0}${\rm rint}_2 \gamma$}%
}}}}
\end{picture}%
\caption{(a) a contour $\gamma$; (b) the blocks of ${\rm ext}\,\gamma$ (red) and ${\rm int}_2 \gamma$ (yellow); (c) the blocks of ${\rm ext}\,\gamma$ (red) and ${\rm rint}_2 \gamma$ (yellow), and the green blocks.}
\label{fig picproof}
\efig
\end{centering}

\end{proof}

The following lemma is motivated by a Peierls argument. The probability that $(x,0)$ is surrounded by a contour can be bounded by a sum over these contours with exponentially small weights.

\begin{lemma}
\label{lem pre-Peierls}
\[
\widetilde\bbP_X(E_x^\circlearrowleft) \leq (2S+1) \sum_{\gamma \in \widetilde\Omega_X \cap E_x^\circlearrowleft} \bigl( \tfrac1n \bigr)^{|\gamma|} \, (2S+1)^{-\frac12 |\gamma|} \, \bigl( 1 + \tfrac1n \bigr)^{-\frac12 n L(\gamma)}.
\]
\end{lemma}

\begin{proof}
We have
\be
\begin{split}
&\sum_{\omega' \in \widetilde\Omega_X \cap E_x^\circlearrowleft} \bigl( \tfrac1n \bigr)^{|\omega'|} (2S+1)^{\caL_X(\omega') - |\omega'|} \\
&= \sum_{\gamma \in \widetilde\Omega_X \cap E_x^\circlearrowleft} \sum_{\omega \in \Omega_{{\rm ext}\,\gamma}} \sum_{\omega' \in \Omega_{{\rm int}_2 \gamma}} \bigl( \tfrac1n \bigr)^{|\gamma \cup \omega \cup \omega'|} (2S+1)^{\caL_X(\gamma \cup \omega \cup \omega') - |\gamma \cup \omega \cup \omega'|} \\
&\leq (2S+1) \sum_{\gamma \in \widetilde\Omega_X \cap E_x^\circlearrowleft} \bigl( \tfrac1n \bigr)^{|\gamma|} (2S+1)^{-\frac12 |\gamma|} \sum_{\omega \in \Omega_{{\rm ext}\,\gamma}} \sum_{\omega' \in \Omega_{{\rm int}_2 \gamma}} \bigl( \tfrac1n \bigr)^{|\omega \cup \omega'|} (2S+1)^{\caL_X(\omega \cup \omega') - |\omega \cup \omega'|}.
\end{split}
\ee
The inequality follows from Lemma \ref{lem last difficulty}.
We now get a lower bound for $\widetilde Z(X)$. Let $A(\gamma) = {\rm int}_1 \gamma \setminus {\rm rint}_2 \gamma$; since $|A(\gamma)| = \frac12 n L(\gamma)$, we have
\be
\begin{split}
\widetilde Z(X) &= \sum_{\omega \in \widetilde\Omega_X} \bigl( \tfrac1n \bigr)^{|\omega|} \, (2S+1)^{\caL_X(\omega) - |\omega|} \\
&\geq \sum_{\omega \in \Omega_{{\rm ext}\,\gamma}} \sum_{\omega' \in \Omega_{{\rm rint}_2 \gamma}} \sum_{\omega'' \in \Omega_{A(\gamma)}} \bigl( \tfrac1n \bigr)^{|\omega \cup \omega' \cup \omega''|} \, (2S+1)^{\caL_X(\omega \cup \omega' \cup \omega'') - |\omega \cup \omega' \cup \omega''|}.
\end{split}
\ee
We restrict the sum over $\omega''$ so that double bars are within single blocks of $A(\gamma)$, that is, $\omega'' \subset A(\gamma)$. Then $\caL_X(\omega \cup \omega' \cup \omega'') = \caL_X(\omega \cup \omega') + |\omega''|$; this holds true because $\omega$ contains only simple loops, and no contours. The sum over $\omega''$ gives $(1 + \frac1n)^{|A(\gamma)|}$ and we obtain the estimate claimed in the statement.
\end{proof}

\begin{proposition}
\label{prop Peierls}
We have
\[
\lim_{n\to\infty} \bbP_{\beta,\ell,n}(E_x^\circlearrowleft | \Omega_{\ell,n}^\alpha) \leq \frac{64}{(2S+1)^{3/2}} + \frac{128}{(2S+1)^2} + \frac1{12} \sum_{k\geq7} \frac{(k+1) 4^k}{(2S+1)^{\frac12 k - 1}}.
\]
\end{proposition}

\begin{proof}
Starting with the upper bound in Lemma \ref{lem pre-Peierls}, we explicitly perform the sum over the external contour $\gamma$ that surrounds $(x,0)$. Let us assign a direction to the loops, namely, they go ``up" on top of the sites $\{-\ell+1,-\ell+3, \dots, \ell-1\}$ and ``down" on top of the sites $\{-\ell+2,-\ell+4,\dots,\ell\}$. Let $y$ be the leftmost site where the loop crosses the $t=0$ line and moves up. We can choose $\gamma$ by summing over numbers $t_1,\dots,t_k$ of time intervals $\frac1n$ between jumps (not counting the ends of the blocks of $X$), and by choosing the $k-2$ directions of the jumps (the first and last directions are forced). Then
\be
\begin{split}
\sum_{\gamma \in \widetilde\Omega_X \cap E_x^\circlearrowleft} &\bigl( \tfrac1n \bigr)^{|\gamma|} \, (2S+1)^{-\frac12 |\gamma|} \, \bigl( 1 + \tfrac1n \bigr)^{-\frac12 n L(\gamma)} \\
&\leq \sum_{y \leq x} \sum_{k \geq |x-y|} 2^{k-2} \bigl( \tfrac1n \bigr)^k (2S+1)^{-\frac12 k} \sum_{t_1, \dots, t_k \geq 1} \prod_{i=1}^k \bigl( 1 + \tfrac1n \bigr)^{-\frac12 t_i}.
\end{split}
\ee
Since $y \in \{ -\ell+1, -\ell+3, \dots, \ell-1 \}$, there are less than $\frac13 (k+1)$ possibilities. Each sum over $t_i$ gives $[1 - (1+\frac1n)^{-1/2}]^{-1}$. It gets multiplied by $\frac1n$, which gives 2 in the limit $n\to\infty$. The contribution of contours with $k$ double bars is then less than $\frac1{12} (k+1) 4^k (2S+1)^{1-\frac12 k}$. The smallest contours have $k=5$; they actually have limited entropy, so the contribution of contours with $k=5$ and $k=6$ is less than $2^{k+1} (2S+1)^{1-\frac12 k}$. This gives the bound of the lemma.
\end{proof}

\begin{proof}[Proof of Theorem \ref{thm dimers}]
From Proposition \ref{prop AN}, Lemma \ref{lem simple configs}, and Proposition \ref{prop surround}, we have for $x \in \{-\ell+3, \dots, \ell-1 \}$ that
\be
\langle P_{x,x+1}^{(0)} \rangle_{\infty,\ell} - \langle P_{x-1,x}^{(0)} \rangle_{\infty,\ell} \geq \bigl( 1 - \tfrac1{(2S+1)^2} \bigr) \inf_\alpha \lim_{\beta\to\infty} \lim_{n\to\infty} \bigl[ 1 - 2 \bbP_{\beta,\ell,n}(E_x^\circlearrowleft | \Omega_{\ell,n}^\alpha) \bigr].
\ee
The result follows from Proposition \ref{prop Peierls}. Indeed, the upper bound is finite when $S > \frac{15}2$ and it is decreasing in $S$; it is equal to $\frac12$ at $S=39.2$.
\end{proof}

\begin{proof}[Proof of Theorem \ref{thm exp decay}]
In a similar fashion as Proposition \ref{prop AN}, we can show that
\be
\langle S_x^i S_y^j \rangle_{\beta,\ell} \rangle = \tfrac13 S (S+1) \delta_{i,j} (-1)^{x-y} \bbP_{\beta,\ell,n}(x \leftrightarrow y).
\ee
This formula was explicitly derived in \cite{AN} and \cite{Uel}. For $x$ and $y$ to be connected, if $|x-y| \geq 2$, there must exists a contour $\gamma$ that surrounds $x$ and such that $|\gamma|$ is greater than $|x-y|$. By Lemmas \ref{lem contour expression}, \ref{lem special Peierls}, and \ref{lem pre-Peierls}, we get
\be
\e{|x-y|/\eta} \bbP_{\beta,\ell,n} (x \leftrightarrow y | \Omega_{\ell,n}^\alpha) \leq (2S+1) \sum_{\gamma \in \tilde\Omega_X \cap E_x^\circlearrowleft} \bigl( \e{1/\eta} \tfrac1n \bigr)^{|\gamma|} (2S+1)^{-\frac12 |\gamma|} \bigl( 1 + \tfrac1n \bigr)^{-\frac12 n L(\gamma)}.
\ee
We can adapt the proof of Proposition \ref{prop Peierls}. The only difference is a factor $\e{k/\eta}$, so we get
\be
\lim_{n\to\infty} \e{|x-y|/\eta} \bbP_{\beta,\ell,n}(x \leftrightarrow y | \Omega_{\ell,n}^\alpha) \leq \tfrac1{12} \sum_{k \geq |x-y|} \frac{(k+1) (4 \e{1/\eta})^k}{(2S+1)^{\frac12 k - 1}}.
\ee
The series converges when $S > \frac{15}2$ and $\eta$ is large enough. The theorem follows from Lemma \ref{lem simple configs}.
\end{proof}

\section{Discussion}\label{sec discussion}

For the spin chain with Hamiltonian \eq{HamSN}, we established the existence of dimerization when $S\geq 40$: in the thermodynamic
limit the model has at least two 2-periodic ground states in which the translation invariance is broken. This follows directly from Theorem
\ref{thm dimers}. We do not expect that these models
have other ground states. In particular, based on what has been shown for the antiferromagnetic $XXZ$ chain \cite{datta:2002},
it seems unlikely that domain-wall ground states exist for this model. We also proved that in the two ground states we constructed, the 
$SU(2)$-invariance remains unbroken.

As stated in Theorem \ref{thm exp decay}, our proof of dimerization also  implies exponential decay of correlations in the ground states. 
Following the arguments of Kennedy and Tasaki \cite{kennedy:1992}, exponential decay implies a spectral gap in this setting. 

One may also ask about the stability of the dimerization under small translation-invariant perturbations of the interaction.
Since the model is not frustration free and involves translation symmetry breaking, the result of Michalakis and Zwolak \cite{michalakis:2013} 
does not apply but we expect that the random loop representation can be used as a starting point for studying perturbations 
of the model using traditional cluster expansion techniques as is done, e.g., in \cite{kennedy:1992}. We have not pursued this possibility
as it is beyond the scope of this work, but establishing stability under arbitrary, uniformly bounded, and sufficiently small perturbation of the interactions, would certainly be un important contribution to understanding the phase diagram of quantum spin chains \cite{affleck:1987}.

\section*{Acknowledgments}

The research reported in this article was completed during the workshop {\em Many-Body Quantum Systems and Effective Theories} at the Mathematisches Forschungsinstitut Oberwolfach, September 11--17, 2016. We thank the organizers and the institute for the great hospitality and the stimulating program. One referee made very useful comments. The work of BN was supported in part by the National Science Foundation under Grant DMS-1515850. 


{
\renewcommand{\refname}{\small References}
\bibliographystyle{symposium}

}

\end{document}